\documentclass[10pt,a4paper]{amsart}

\usepackage{amssymb}
\usepackage[final]{showkeys}
\usepackage{microtype}
\usepackage{color}
\usepackage{graphicx}
\usepackage[hmargin=3cm,vmargin={5cm,4cm}]{geometry}

\newtheorem{te}{Theorem}[section]

\newtheorem{os}[te]{Remark}

\numberwithin{equation}{section}

\allowdisplaybreaks

\begin{document}

\title[]{A note on generalized fractional diffusion equations on Poincar\'e half plane.}
\author{R. Garra}
\author{F. Maltese}
\author{E. Orsingher}

\begin{abstract}
In this paper we study generalized time-fractional diffusion equations on the Poincar\'e half plane $\mathbb{H}_2^+$. The 
time-fractional operators here considered are fractional derivatives of a function with respect to another function, that 
can be obtained by starting from the classical Caputo-derivatives essentially by means of a deterministic change of variable.
We obtain an explicit representation of the fundamental solution of the generalized-diffusion equation on $\mathbb{H}_2^+$ 
and provide a probabilistic interpretation related to the time-changed hyperbolic Brownian motion. 
We finally include an explicit result regarding the non-linear case admitting a separating variable solution.

\bigskip

\end{abstract}
\maketitle

\section{Introduction}

In this paper we study generalized time-fractional diffusion equations on the hyperbolic 
Poincar\'e half-plane 
$$\mathbb{H}^{+}_2=\left\{ (x,y) \in \mathbb{R}^2\bigg| y>0\right\}.$$
The generalization here considered is based on the application of time-fractional 
derivatives of a function with respect to another function (see \cite{Almeida} for the definition and main properties), an interesting 
approach that permits us to take into account both time-varying coefficients and memory effects
(see e.g. \cite{aml} for a physical discussion about this).
In the previous paper \cite{Lao} the authors studied for the first time the time-fractional diffusion equation on the hyperbolic space involving the classical Caputo derivative.
Moreover, in the more recent paper \cite{mirko}, an interesting probabilistic interpretation of the fundamental solution of the time-fractional telegraph-type equation on hyperbolic spaces has been provided. In particular, a relevant connection with time-changed hyperbolic Brownian 
motions has been proved. \\
The main aim of this paper is to provide a rigorous analysis of the generalized time-fractional
diffusion equation on the hyperbolic space $\mathbb{H}^{+}_2$. We find an explicit representation of the fundamental solution by means of the method
of separation of variables.
Moreover, we obtain a probabilistic interpretation of the related stochastic process as a time-changed hyperbolic Brownian motion. 
In the first part of the paper we provide some necessary preliminaries about the Poincar\'e half-plane and the definition and basic properties of the fractional operators here considered.
We decided to provide detailed preliminaries about the Poincar\'e half-plane since many non-trivial computations are involved and we think that this short guide can be of help for the reader.\\
Then, we analyze the generalized time-fractional diffusion equation on $\mathbb{H}^{+}_2$ 
providing the representation of the fundamental solution and the related probabilistic meaning.
Finally, we also consider a nonlinear generalized time-fractional diffusion equation on $\mathbb{H}^{+}_2$ admitting a solution obtained 
by means of the method of separation of variables.\\
Few papers are devoted to the analysis of time-fractional diffusive equations on hyperbolic 
spaces, in our view, together with the previous papers \cite{mirko} and \cite{Lao}, this can be another step to develop this new and interesting topic. 

\section{Preliminaries}

\subsection{A short survey on hyperbolic geometry}

We here give some necessary mathematical preliminaries about the model of the Poincair\'e half-plane i.e the set 
$\mathbb{H}^{+}_2=\left\{ (x,y) \in \mathbb{R}^2|y>0\right\}$ 
with the following metric
\begin{equation}\label{0}
ds^2=\frac{dx^2+dy^2}{y^2}
\end{equation}

 First of all, in order to characterize the geometry of the Poincar\'e half-plane, we study the form of the geodesics, by using the variational principle. 

We consider the family of curves in the hyperbolic plane passing through two given points $(x_1,y_1) $ and $(x_2,y_2)$   in their parametric representation i.e.

\begin{equation}\label{1}
\gamma=\left\{ (x(t),y(t))\bigg| t_1  \leq t \leq t_2 \right\},
\end{equation}

where $t_1$ and  $t_2$ are such that $(x(t_1),y(t_1))=(x_1,y_1) $ and $(x(t_2),y(t_2))=(x_2,y_2)$.

The length of this curve in the hyperbolic plane is
\begin{equation}\label{2}
\mathcal{L} (\gamma)=\int_{t_1}^{t_2} \frac {\sqrt{x'(t)^{2}+y'(t)^{2} }}{y(t)}dt.
 \end{equation}

We can simplify this expression by restricting ourselves to the family of parametric curves  of \eqref{1} to curves with Cartesian parameterization i.e

\begin{equation}
\gamma=\left\{ (x,y(x)) | x_1  \leq x \leq x_2 \right\}.
\end{equation}

In this case, the integral \eqref{2} becomes 
\begin{equation}\label{3}
\mathcal{L} (\gamma)=\int_{x_1}^{x_2}{ \frac {\sqrt{1+y'(x)^{2} }}{y(x)}dx}.
\end{equation}

We can consider an arbitrary function $w(x)$ as $w(x)=y(x)+\epsilon h(x)$ with $\epsilon \geq 0$ and $h(x)$ is such that $h(x_1)=h(x_2)=0$ .

So the length of curve $(x,w(x))$ applying \eqref{3} becomes
\begin{equation}
\mathcal{L} (\gamma)=l(\epsilon)=\int_{x_1}^{x_2}{ \frac {\sqrt{1+(y'(x)+\epsilon h'(x))^{2} }}{y(x)+\epsilon h(x)}\,dx}.
\end{equation}

The geodesic curve is associated with a minimum point with respect to the $\epsilon$ variable of the function $l(\epsilon)$ for which is satisfied the condition
\begin{equation}\label{40}
\frac {dl}{d\epsilon}\bigg|_{\epsilon=0}=0.
\end{equation}
 
By direct computation we have that 

\begin{align}	
\nonumber &\frac {dl}{d\epsilon}\bigg|_{\epsilon=0}=
 \int_{x_1}^{x_2}{\frac {d}{d\epsilon}\left( \frac {\sqrt{1+(y'+\epsilon h')^{2} }}{y+\epsilon h}\right)|_{\epsilon=0}dx}=\\
\nonumber &= \int_{x_1}^{x_2}{\left(- \frac {h\sqrt{1+(y'+\epsilon h')^{2} }}{(y+\epsilon h)^2}+\frac {h'(y'+\epsilon h')}{(y+\epsilon h)\sqrt{1+(y'+\epsilon h')^2}}\right)|_{\epsilon=0}dx}=\\
\nonumber & =\int_{x_1}^{x_2}\left({- \frac {h\sqrt{1+y'^2 }}{y^2}+\frac {h'y'}{y\sqrt{1+y'^2}} }\right) dx
=\int_{x_1}^{x_2}- \frac {h\sqrt{1+y'^2 }}{y^2}dx+\left[\frac {hy'}{y\sqrt{1+y'^2}}\right]_{x_1}^{x_2}\\
\nonumber & -\int_{x_1}^{x_2}h \frac {d}{dx}\frac{y'}{y\sqrt{1+y'^2}}dx \\ 
\nonumber &\mbox{(since the function $h(x)$ is such that $h(x_1)=h(x_2)=0$)}\\
 &=\int_{x_1}^{x_2}\left(- \frac {\sqrt{1+y'^2 }}{y^2}- \frac {d}{dx}\frac{y'}{y\sqrt{1+y'^2}}\right)h dx \label{4}
\end{align}

By  \eqref{40}, the integral \eqref{4} must be equal to zero for all
functions $h$ and therefore we have that \\

\begin{equation}
\nonumber - \frac {\sqrt{1+y'^2 }}{y^2}- \frac {d}{dx}\frac{y'}{y\sqrt{1+y'^2}}=0
\end{equation}
and therefore
 \begin{equation}
- \frac{1 }{y^2\sqrt{1+y'^2 }}- \frac{y''}{y\left(\sqrt{1+y'^2}\right)^3}=0.
\end{equation}

We finally obtain that
 \begin{equation} \label{6}
1+  \frac{d }{dx}(yy')=0.
\end{equation}
   Integrating twice \eqref{6}  we have the following equation     
     \begin{equation}
x^2 +y^2-2cx-2d=0
 \end{equation}
 which is in $\mathbb{H}^{+}_2$  the equation of the semi-circles with an arbitrary center on the x-axis with arbitrary radius. \\
 Other geodesic curves  in the Poincair\'e half-plane are the half-lines parallel to the y-axis, that emerge if $x_1=x_2$. 
  
 \bigskip
 
 In order to derive the expression of the Laplacian in $\mathbb{H}^{+}_2$, let us introduce the hyperbolic coordinates. First of all, we need introduce the geometric center of the upper Poincar\'e half-plane which is the point $(0,1)$ and consider an arbitrary point of Cartesian coordinates $ (x,y) $ in $\mathbb{H}^{+}_2$ .
 The hyperbolic coordinates associated with this point are $(\eta , \alpha)$  where $\eta$ is the hyperbolic distance between $(0,1)$ and $(x,y)$, i.e., the length according to the metric \eqref{0} of the arc of geodesic that passes through $(0,1)$ and $ (x,y)$, which is an arc of semi-circumference if $(x,y)$ is not on the y-axis.
 While $\alpha$ is the angle formed by the tangent to that semi-circumference in $(0,1)$ and passing through $(x,y)$. \\
 As hyperbolic coordinates are defined, important relationships can be obtained from the transition to hyperbolic coordinates to the Cartesian coordinates, for example from the $\alpha $ coordinate to the Cartesian coordinates we can get a relationship that starts from the equation of the geodesic that passes through the origin $(0,1)$ and the point $(x,y)$
 
     \begin{equation}
(x-\tan \alpha)^2 + y^2=\frac{1}{\cos^2  \alpha} 
   \end{equation}
and therefore 
\begin{equation}\label{7}
 \tan \alpha =\frac{x^2+y^2-1}{2 x}. 
 \end{equation}
  
  Finally, the relationship between the $\eta$ coordinate and the Cartesian coordinates
 is given by  \label{8}
  \begin{equation}
  \cosh \eta =\frac{x^2+y^2+1}{2 y}. 
  \end{equation}

  Starting from the relations \eqref{7} and \eqref{8} we can get the change of coordinates from cartesian coordinates $(x,y)$ to hyperbolic coordinates $(\eta,\alpha)$ in $\mathbb{H}^{+}_2$    
  \begin{equation}\label{coo}
  \begin{cases}
  x =\displaystyle\frac{\cos \alpha \sinh \eta}{\cosh \eta - \sinh \eta \sin \alpha}, \quad \eta >0, \ 0< \alpha<2\pi \\                   
   y=\displaystyle\frac{1}{\cosh \eta - \sinh \eta \sin \alpha}
  \end{cases}
  \end{equation}
  
  We are now able to derive the expression of the Laplacian operator in hyperbolic coordinates.

  First of all, we observe that the Poincar\'e upper half-plane is a Riemannian manifold with the following metric tensor
  
  $$ \textbf{g}=  \left( \begin{array}{cc}
  
  \frac{1}{y^2} & 0 \\
  0 & \frac{1}{y^2} 
  \end{array} \right),  \\\\  $$
  
  In general, on a Riemannian manifold with a metric tensor  \textbf{g}, the Laplacian is
  given by
  
   \begin{equation}
  \Delta f= \frac{1}{\sqrt{|\textbf{g}|}} \sum_{i=1}^n \partial_{i}(\sqrt{|\textbf{g}|} \sum_{j=1}^n g^{ij} \partial_j f),
  \end{equation}
  
  where $|\textbf{g}|$  is the determinant of the metric tensor and the elements $g^{ij}$ are the components of the inverse matrix of \textbf{g} and $n$ is the dimension of the manifold.\
  
 By observing that the inverse matrix  $\textbf{g}^{-1}$ and $|\textbf{g}|$  are respectively given by
  
   \begin{equation}
   \textbf{g}^{-1}=  \left( \begin{array}{cc}     
  y^2 & 0 \\
  0 & y^2 
  \end{array} \right), \quad |\textbf{g}|=\frac{1}{y^4},
  \end{equation}
  we have that in this case $\Delta f$ becomes
  
   \begin{align}
   \nonumber &\Delta f= \frac{1}{\sqrt{\frac{1}{y^4} }} \left[ \frac{\partial}{\partial x} \left (\sqrt{\frac{1}{y^4}} y^2  \frac{\partial}{\partial x} f\right)+ \frac{\partial}{\partial y} \left (\sqrt{\frac{1}{y^4}} y^2  \frac{\partial}{\partial y} f\right)\right] =\\
    &= y^2 \left[ \frac{\partial}{\partial x}\left ( \frac{\partial}{\partial x} f\right)+ \frac{\partial}{\partial y}\left( \frac{\partial}{\partial y} f\right)\right]= y^2 \left( \frac{\partial^2}{\partial x^2}+ \frac{\partial^2}{\partial y^2}\right)f, \quad x \in \mathbb{R}, y>0 \label{coos}
  \end{align}
  
  By applying \eqref{coo} and \eqref{coos} we can derive the expression of the
  Laplacian in hyperbolic coordinates (see \cite{Lao} Theorem 2.1 for detailed calculations) that is given by the differential operator
  \begin{equation}
  \frac{1 }{\sinh\eta} \frac{\partial }{\partial \eta}\left(\sinh\eta\frac{\partial }{\partial \eta}\right)+ \frac{1 }{\sinh^2\eta}\frac{\partial^2 }{\partial \alpha^2}, \quad 0<\alpha<2\pi, \ \eta > 0.
  \end{equation}

\subsection{Fractional derivatives of a function with respect to another function}

Fractional derivatives of a function with respect to another function have been considered in the classical monograph by Kilbas et al. \cite{kilbas} (Section 2.5) and recently studied by Almeida in \cite{Almeida} that has provided the Caputo-type regularization of the existing definition and some interesting properties. Starting from this paper, this topic has gained interest both for mathematical reasons (see e.g. \cite{new}) and for physical applications (e.g. \cite{aml} and
the references therein). 
The utility of these generalized fractional operators in the applications is represented 
by the fact that they are essentially obtained by a deterministic time-change and permits us to 
take into account both time-variable coefficients and memory effects. Moreover, this class of operators include as special cases classical well-known time-fractional derivatives (for example, fractional derivatives in the sense of Hadamard, or Erd\'elyi-Kober).\\
Here we recall the basic definitions and properties for the reader's convenience.

Let $\nu>0$, $f\in C^1([a,t])$ an increasing function such that $f'(t)\neq 0$ in $[a,t]$, the fractional integral of a function $g(t)$ 
with respect to another function $f(t)$ is given by 
\begin{equation}
\left(I^{\nu,f}_{a^+}g\right)(t):=\frac{1}{\Gamma(\nu)}\int_a^t f'(\tau)
(f(t)-f(\tau))^{\nu-1}g(\tau)d\tau.
\end{equation}
Observe that for $f(t) = t^{\beta}$ we recover the definition of Erd\'elyi-Kober fractional integral recently applied, for example, in connection with the Generalized Grey Brownian Motion \cite{gianni}. For simplicity hereafter we will consider $a = 0$ (as usual) and suitable functions 
$f$ such that $f(0) = 0$. All the results can be simply generalized.

The corresponding Caputo-type evolution operator (see \cite{Almeida}) for $0<\nu<1$ is given by
\begin{align}\label{2.1}
\left(\mathcal{O}^{\nu, f}g\right)(t):=&\frac{1}{\Gamma(1-\nu)}\int_0^t 
(f(t)-f(\tau))^{-\nu}\frac{d}{d\tau}g(\tau)d\tau\\
& = I_{0^+}^{1-\nu,f}\left(\frac{1}{f'(t)}\frac{d}{dt}\right) g(t).
\end{align}
For the general case $\nu \in \mathbb{R}$ we refer to \cite{Almeida}. In this paper we are interested to the case $0<\nu<1$ interpolating
as a limit case the ordinary first order derivative, while the higher order cases can be treated in a similar way. We have used the symbol $\mathcal{O}^{\nu, f}(\cdot)$ in order to underline the generic integro-differential nature of the time-evolution operator, depending on the choice of the function $f(t)$ and the real order $\nu$.

A relevant property of the operator (\ref{2.1}) is that if $g(t) =(f(t))^{\beta-1}$ with $\beta>1$, then (see Lemma 1 of \cite{Almeida}) 
\begin{equation}\label{2.2}
\left(\mathcal{O}^{\nu, f} g\right)(t) = \frac{\Gamma(\beta)}{\Gamma(\beta-\nu)}(f(t))^{\beta-\nu-1}.
\end{equation}
Indeed, by direct calculation we have that
\begin{equation}
\nonumber \left(\mathcal{O}^{\nu, f} f^{\beta-1}\right)(t) = \frac{\beta-1}{\Gamma(1-\nu)}\int_{0}^{t}(f(t)-f(\tau))^{-\nu} f'(\tau)(f(\tau))^{\beta-2}d\tau
\end{equation}
and taking $y = f(\tau)/f(t)$ we have that
\begin{align}
\nonumber \left(\mathcal{O}^{\nu, f} f^{\beta-1}\right)(t)& = \frac{\beta-1 \ f^{\beta-1-\nu}}{\Gamma(1-\nu)}\int_{0}^{1}(1-y)^{-\nu} y^{\beta-2}dy\\
\nonumber &= \frac{\Gamma(\beta) \ f^{\beta-1-\nu}(t)}{\Gamma(\beta-1)\Gamma(1-\nu)}\frac{\Gamma(1-\nu)\Gamma(\beta-1)}{\Gamma(\beta-\nu)}\\
\nonumber & = \frac{\Gamma(\beta)}{\Gamma(\beta-\nu)}(f(t))^{\beta-\nu-1}.
\end{align}
Therefore, the composite Mittag-Leffler function 
\begin{equation}
g(t)= E_\nu(\lambda(f(t))^\nu)
\end{equation}
is an eigenfunction of the operator $\mathcal{O}^{\nu, f}$, when $\nu\in(0,1)$ and $f$ is a well-behaved function such that $f(0)=0$. This means that
\begin{equation}
\mathcal{O}^{\nu, f}E_\nu(\lambda(f(t))^\nu) = \lambda E_\nu(\lambda(f(t))^\nu). 
\end{equation}

\section{Generalized linear and nonlinear fractional diffusion on Poincar\'e half-plane}

\subsection{The linear case}

In a previous paper \cite{Lao}, the authors considered the following diffusion-type equation on $\mathbb{H}_2^+$
\begin{equation}\label{lao}
\frac{\partial^\beta}{\partial t^\beta} u(\eta,t) =  \frac{1}{\sinh\eta}\left(\frac{\partial}{\partial \eta}\sinh\eta \frac{\partial}{\partial \eta}\right)u(\eta,t), \quad \beta \in (0,1), 
\end{equation}
where $\frac{\partial^\beta}{\partial t^\beta}$ is a fractional derivative of order $\beta$ in the sense of Caputo. We here analyze the more general case involving the fractional derivative w.r.t. another function. First of all, we have the following result
\begin{te}
Let be $f\in L^1[0,t]$ such that $f(0)=0$, the fundamental solution for the generalized time-fractional diffusion equation 
\begin{equation}\label{dome}
\left(\mathcal{O}^{\beta, f} u\right)(\eta,t) = \frac{1}{\sinh\eta}\left(\frac{\partial}{\partial \eta}\sinh\eta \frac{\partial}{\partial \eta}\right)u(\eta,t) 
\end{equation}
is given by 
\begin{equation}
u(\eta,t) = \frac{2}{\pi}\int_0^\infty xE_{\beta}\left(-\frac{f(t)^\beta}{4}-x^2f(t)^\beta\right)dx\int_\eta^\infty d\varphi\frac{\sin(x\varphi)}{\sqrt{2\cosh\varphi-2\cosh\eta}}.
\end{equation}
\end{te}

\begin{proof}
We find the solution by means of the separation of variables and transform the Laplacian operator by using the change of variable
$y = \cosh \eta$ which leads to  
$$u(y,t) = F(y)\cdot T(t)$$
and therefore we get
\begin{align}
& \left(\mathcal{O}^{\beta, f} T\right) = -\omega T,\\
&(y^2-1)F''+2y F'+\omega F = 0.
\end{align}
The solution of the first equation is given by 
\begin{equation}
T(t, \omega)= E_{\beta, 1}(-\omega f(t)^\beta).
\end{equation}
The spatial part of the solution remains the same as in the classical hyperbolic diffusion equation and we refer to \cite{Lao} for the details.
\end{proof}

\begin{os}
Observe that for $f(t) = t$ and $\beta = 1$ we recover the transition function of the hyperbolic Brownian motion, firstly studied by Gertsenshtein and Vasiliev in \cite{gert}.\\
Moreover, for $f(t)= t$ and $\beta \in (0,1)$ we recover the results obtained in \cite{Lao}.
\end{os}

Let us introduce the process 
$$T^\beta(f(t))= B^{hp}(\mathcal{L}^\beta(f(t))),$$
where $B^{hp}$ is the hyperbolic Brownian motion in $\mathbb{H}_2^+$ independent from 
$\mathcal{L}^\beta(t)$ which is the inverse of the stable subordinator $H^\beta(t)$, that is
$$\mathcal{L}^\beta(t)= \inf\{s>0: H^\beta(s)\geq t\}, \quad \beta \in (0,1).$$
We have the following 
\begin{te}
The distribution $p(x,t)$ of the process $T^\beta(f(t))$ coincides with the fundamental solution
of the equation \eqref{dome}.
\end{te}

\begin{proof}
We observe that, by means of the deterministic time-change $f(t)\rightarrow t$, we can essentially go back to a time-fractional diffusion equation involving the Caputo derivative.
Then, by means of the time-Laplace transform method, it can be proved that the fundamental solution of \eqref{lao} coincides with the distribution of the process $T^\beta(f(t))$.
\end{proof}

Observe that this paper is devoted to diffusive models in the Poincar\'e half-space 
$\mathbb{H}_2^+$ but the generalizations to $\mathbb{H}_n^+$ can be obtained in a similar way 
from the probabilistic point of view and will be the object of a further detailed analysis.\\

Finally, by means of similar methods, we can generalize the recent results obtained in \cite{mirko} about time-fractional telegraph-type equations in $\mathbb{H}_n$. In particular, we have that

\begin{te}
The distribution of the composition 
\begin{equation}
\mathcal{T}^\beta(t)= B^{hp}(L^\beta(f(t))),
\end{equation}
where $$L^\beta(t) = \inf\{s>0: H_1^{2\beta}(s)+(2\lambda)^{1/\beta}H_2^\beta(s)\geq t\},$$
and $H_1^{2\beta}$, $H_2^\beta$ are independent stable subordinators (with $\beta\in(0,1/2)$, coincides with the fundamental solution of the equation
\begin{equation}\label{domes}
\left(\mathcal{O}^{2\beta, f} u\right)(\eta,t)+2\lambda\left(\mathcal{O}^{\beta, f} u\right)(\eta,t) = \frac{1}{\sinh\eta}\left(\frac{\partial}{\partial \eta}\sinh\eta \frac{\partial}{\partial \eta}\right)u(\eta,t), \quad \beta \in (0,1/2).
\end{equation}
\end{te}
The main idea for the proof is essentially the same of the previous theorem. 
The result can be generalized to a multi-term fractional equation involving a finite number of fractional derivatives w.r.t. another function of order less than one (see \cite{bruno}).

\subsection{The nonlinear case.}

We recall that the construction of the explicit representation of the fundamental solution of the linear diffusion equation is based (also in the fractional case) on the classical method of separation of variables. Inspired by this, we observe that particular solutions for nonlinear equations 
can be constructed by the generalized method of separation of variables (see \cite{polyanin}).
Based on this idea, a final result on non-linear diffusive equation in $\mathbb{H}_2$ is here considered. 
\begin{te}
The generalized time-fractional nonlinear diffusive equation in $\mathbb{H}_2^+$
\begin{equation}\label{domen}
\left(\mathcal{O}^{\beta, f} u\right)(\eta,t) = \frac{1}{\sinh\eta}\left(\frac{\partial}{\partial \eta}\sinh\eta \frac{\partial}{\partial \eta}\right)u^n(\eta, t)-u(\eta,t), \quad n > 0 , (\eta,t)\in \mathbb{R}^+\times \mathbb{R}^+
\end{equation}
admits as a particular solution 
\begin{equation}
u(\eta, t) = g(\eta)\cdot E_{\beta}\left(-(f(t))^\beta\right),
\end{equation}
where $g(\eta)$ is such that 
$\frac{d g^n}{d\eta}=\frac{1}{\sinh \eta}.$
\end{te}

\begin{proof}
We first search a solution by the generalized separation of variables in the simple form
$$u(\eta, t) = r(t)\cdot g(\eta).$$
We observe that if $g(\eta)$ is such that
$$\frac{d g^n}{d\eta}=\frac{1}{\sinh \eta},$$
 then 
$$\frac{1}{\sinh\eta}\left(\frac{\partial}{\partial \eta}\sinh\eta \frac{\partial}{\partial \eta}\right)u^n(\eta, t)=0$$
and therefore by substitution we have that
$$g(\eta)\left(\mathcal{O}^{\beta, f}r\right)(t)=-g(\eta)r(t)$$
and therefore $r(t) = E_{\beta}\left(-(f(t))^\beta\right)$. 
\end{proof}
The study of nonlinear diffusive equations in $\mathbb{H}_2^+$ is not the main object of this paper, but we observe that by starting from this simple result, it is possible to construct exact solutions for many different classes of generalized time-fractional nonlinear equations in 
$\mathbb{H}_2^+$, a completly new topic of research.

\end{document}